\documentclass[11pt]{article}

\usepackage{natbib}
\usepackage{fullpage}
\usepackage{amsmath}
\usepackage{amsthm}
\usepackage{amsfonts}
\usepackage{graphicx}

\usepackage{graphicx}
\usepackage{color}
\usepackage{subcaption}

\newtheorem{theorem}{Theorem}[section]

\newtheorem{lemma}[theorem]{Lemma}

\newtheorem{corollary}[theorem]{Corollary}
\newtheorem{definition}[theorem]{Definition}

\renewcommand{\vec}[1]{\boldsymbol{#1}}
\newcommand{\R}{\mathbb{R}}
\newcommand \reals {\mathbb{R}}
\newcommand{\alg}{\mathcal{A}}
\newcommand{\numfunctions}{T}
\newcommand{\E}{\mathop{\mathbb{E}}}

\newcommand{\crit}{c}
\newcommand{\comment}[1]{}

\newcommand{\dmin}{\operatorname{d_{\rm min}}}
\newcommand{\dmax}{\operatorname{d_{\rm max}}}

\newcommand{\deltamin}{\Delta_\text{min}}
\newcommand{\deltamax}{\Delta_\text{max}}

\newcommand{\sample}{\mathcal{S}}

\newcommand{\pdim}{\textnormal{Pdim}}
\newcommand{\dist}{\mathcal{D}}
\newcommand{\distr}{\mathcal{D}}

\newcommand \OPT {\mathrm{OPT}}
\newcommand \vol {\operatorname{Vol}}

\newcommand \expect {\mathbb{E}}
\newcommand \partition {\mathcal{P}}
\newcommand \argmax {\operatorname*{argmax}}

\newcommand{\Pbi}{\Pi}
\newcommand{\algfamily}{\mathcal{A}}

\newcommand \norm [1] {\Vert#1\Vert}

\newcommand{\params}{\mathcal{P}}
\newcommand{\configs}{\mathcal{P}}

\begin{document}

\title{Data-driven Algorithm Design\footnote{Chapter $29$ of the book ``Beyond the Worst-Case Analysis of Algorithms''\citep{tim20}.}
}
\author{Maria-Florina Balcan\\ School of Computer Science\\Carnegie Mellon University\\\texttt{ninamf@cs.cmu.edu}}

\date{}

\maketitle

\begin{abstract}
Data driven algorithm design is an important aspect of modern data science and algorithm design. Rather than using off the shelf algorithms that only have worst case performance guarantees, practitioners often optimize over large families of parametrized algorithms and tune the parameters of these algorithms using a training set of problem instances from their domain to determine a configuration with high expected performance over future instances. However, most of this work comes with no performance guarantees. The challenge is that for many combinatorial problems of significant importance including partitioning, subset selection, and alignment problems, a small tweak to the parameters can cause a cascade of changes in the algorithm's behavior, so the algorithm's performance is a discontinuous function of its parameters.

In this chapter, we survey recent  work that helps put data-driven combinatorial algorithm design on firm foundations. We provide strong computational and statistical performance guarantees, both for the batch and online scenarios where a collection of typical problem instances from the given application are presented either all at once or in an online fashion, respectively.
\end{abstract}

\section{Motivation and Context}
The classic approach to designing and analyzing combinatorial algorithms (that has been the backbone of algorithmic research and applications since the field's inception) assumes that the algorithm we design for a given problem will be used to solve worst-case instances of the problem, about which the algorithm has absolutely no information at all.  The typical performance guarantees we aim for, in this classic framework, require that the algorithm we design must succeed even for solving just a one time, worst case instance of the underlying algorithmic problem. While ideal in principle, for many problems such worst case guarantees are often weak. Moreover, for many problems, empirically, different methods work better in different settings, and there are often large, even infinite parametrized families of methods that one could try to use. Consequently, rather than using off the shelf algorithms that have weak worst case guarantees, practitioners often employ a data-driven algorithm design approach; specifically, given an application domain, they use machine learning and instances of the problem coming from the specific domain to learn a method that works best in that domain. This idea has long been used in practice in various communities, including artificial intelligence~\citep{eric01,Xu08:SATzilla}, computational biology~\citep{DBK18}, and auction design~\citep{sand03}. However, so far, most of this work has come with no performance guarantees.

In this chapter we survey recent work that provides formal guarantees for this  data-driven algorithm design approach, by building and significantly expanding on learning theory tools. We discuss  both  the batch and online scenarios where a collection of typical problem instances from the given application are presented either all at once or in an online fashion, respectively.   This includes nearly optimal sample complexity bounds for the batch scenario and no-regret guarantees for the online scenario for a number of important algorithmic families that include classic modules such as greedy,  local search, dynamic programming, and  semi-definite relaxation followed by rounding. These are applicable to a wide variety of combinatorial problems (e.g.,  subset selection, clustering, partitioning, and alignment problems) from diverse domains ranging from data science to computational biology to auction design.  The key technical challenge is that for many of these problems, a small tweak to the parameters can cause a cascade of changes in the algorithm's behavior, so the algorithm's performance is a discontinuous function of its parameters.

At a technical level, this work draws on insights from some of the other approaches on algorithms beyond the worst case, including perturbation stability~\cite{bilulinial}  and  approximation stability~\citep{bbg13}. The motivation here is identical: many important optimization problems are unfortunately provably hard even to approximate well on worst-case instances, so using algorithms with worst case guarantees might be pessimistic.
The key difference is that that line of work aims to articulate specific regularities or stability properties that the input instances might satisfy, and to design algorithms that provably exploit them and overcome worst-case hardness results on instances satisfying them. In addition to providing algorithms with provable guarantees when these stability conditions hold, such analyses suggest interesting families of algorithms to learn over in the data-driven algorithm approach, that are even more broadly applicable (including in scenarios where verifying these properties might be  hard). Indeed, some of the algorithm families we study in this chapter (in the context of clustering problems in particular) are directly inspired by these analyses.

This topic is related in spirit to several widely popular topics in machine learning, including hyperparameter tuning and meta-learning. The key difference here is that we focus on parametric families of functions induced by algorithms for solving discrete optimization problems, which leads to cost functions with sharp discontinuities. This leads to very interesting challenges that require new techniques that help significantly push the boundaries of learning theory as well.

The goals of data-driven algorithm design are similar to those of self-improving algorithms.  The main take away of our chapter is that one can build on and extend tools from learning theory to achieve these goals for a wide variety of algorithmic problems.

\section{Data-driven Algorithm Design via Statistical Learning}
\label{se:dislearn}

~\cite{GR16, GR17} proposed analyzing data-driven algorithm design as a distributional learning problem, by using and extending the classic learning theory models, PAC~\citep{Valiant:acm84} and Statistical Learning Theory~\citep{Vapnik:book98}. In this framework, for a given algorithmic problem, we model an application domain as a distribution over problem instances, and assume that we have access to training instances that are drawn i.i.d from this fixed, but unknown distribution. The formal guarantees we aim for in this framework are generalization guarantees quantifying how many training problem instances are needed to ensure that an algorithm with good performance over the training instances will exhibit good performance on future problem instances. Such guarantees depend on the  intrinsic complexity  of  the search space which in this case is a parametrized family of algorithms for the problem  at  hand, and this intrinsic dimension is quantified using learning theoretic measures.

The challenge, and the reason that theoretical analysis is needed, is that it could be that parameter settings that work well on past instances perform poorly on future instances due to overfitting to the training data.   In particular, even if past and future instances are all drawn i.i.d. from the same probability distribution, if the algorithm family is sufficiently complex, it may be possible to set parameters that capture peculiarities of the training data (or even in the extreme case, memorize specific solutions to training instances), performing well on them without truly performing well on the instance distribution.  Sample complexity analysis provides guarantees on how many training instances are sufficient, as a function of the complexity of the algorithm family, to ensure that with high probability no such overfitting occurs.  Below, we formally describe the problem setup and how overfitting will be addressed through uniform convergence analysis.

\paragraph{Problem Formulation}
We fix an algorithmic problem (e.g., a subset selection problem or a clustering problem) and we  denote by $\Pbi$  the set of problem instances of interest for this problem. We also fix $\mathcal{A}$ a large (potentially infinite) family of algorithms, and throughout this chapter we assume that this family  is
parameterized by a set $\configs \subseteq \R^d$; we denote by $A_{\vec{\rho}}$ the algorithm in $\mathcal{A}$ parametrized by $\vec{\rho}$. We also fix   a utility function $u : \Pbi \times
\configs \to [0,H]$,  where $u(x, \vec{\rho})$ measures the
performance of the algorithm $A_{\vec{\rho}}$  on problem instance $x \in
\Pbi$. We denote by $u_{\vec{\rho}}(\cdot) $ the utility function $u : \Pbi
 \to [0,H]$ induced by $A_{\vec{\rho}}$,  where $u_{\vec{\rho}}(x) = u(x,\vec{\rho})$.  Note that  $u$ is bounded; for example, for cases where  $u$  is related to an algorithm's
 running time, $H$ can be the time-out deadline.

The ``application-specific information'' is modeled by the unknown input distribution $\distr$. The learning algorithm is given $m$ i.i.d. samples $x_1, . . . , x_m \in \Pbi$ from $\distr$,
and (perhaps implicitly) the corresponding performance $u_{\vec{\rho}}(x)$ of each algorithm
$A_{\vec{\rho}} \in \mathcal{A}$ on each input $x_i$.  The learning algorithm uses this information to suggest an
algorithm $A_{\hat{\vec{\rho}}} \in \mathcal{A}$ to use on future inputs drawn from $\distr$.
 We seek learning algorithms
that almost always output an algorithm of $\mathcal{A}$ that performs almost as well as the
optimal algorithm  $A_{{\vec{\rho}}^{\star}}$  for $\distr$ that  maximizes $\E_{x \sim \distr}[u_{\vec{\rho}}(x)]$ over $A_{\vec{\rho}} \in \mathcal{A}$.

\paragraph{Knapsack} As an example, a canonical problem we consider in this chapter is the knapsack problem. A knapsack instance $x$ consists of $n$ items, where each item $i$ has a value $v_i$ and a size $s_i$, together with an overall knapsack capacity $C$. Our goal is to find the most valuable subset of items
whose total size does not exceed $C$. For this problem we analyze a family of greedy algorithms parametrized by a one dimensional set, $\configs = \reals$. For $\rho \in \configs$, the algorithm $A_\rho$ operates as follows. We set the score of item $i$ to be $v_i / s_i^\rho$; then, in decreasing order of score, we add each item to the
knapsack if there is enough capacity left (breaking ties by selecting the item of smallest index).
The
 utility function  $u_\rho(x)=u(x, \rho)$ is defined as the value of the items chosen by
the greedy algorithm with parameter $\rho$ on input $x$.

\paragraph{Uniform Convergence}
To achieve our desired guarantees, we rely on uniform convergence results, which roughly speaking specify how many training instances we need in order to guarantee that with high probability (over the draw of the training set of instances) we have that, uniformly, for all the algorithms in the class $\mathcal{A}$,  their average performance over the sample is additively close to their expected performance on a typical (random) problem instance coming from the same distribution as  the training set.  It is known from empricial processes and learning theory that these uniform convergence results depend on the intrinsic complexity of the family of real-valued utility functions $\{u_{\vec{\rho}}(\cdot)\}_{\vec{\rho}}$. 
In this chapter we consider the pseudo-dimension as a measure of complexity, which roughly speaking quantifies the ability of the class to fit complex patterns.
%

\begin{definition} [Pseudo-dimension]
\label{def-pseudo}
Let $\{u_{\vec{\rho}}(\cdot)\}_{\vec{\rho}}$ be the family of performance measures induced by 
$\mathcal{A}= \{A_{\vec{\rho}}\}_{\vec{\rho}}$ and the utility function $u(x, \vec{\rho})$.
\begin{itemize}
  \item[(a)] Let $\sample= \left\{x_1, \dots, x_m\right\} \subset \Pbi$ be a set of problem instances and let $z_1, \dots, z_m \in \R$ be a set of \emph{targets}.
  We say that $z_1, \dots, z_m$ \emph{witness} the shattering of $\sample$ by
 $\{u_{\vec{\rho}}(\cdot)\}_{\vec{\rho}}$ if for all subsets $T \subseteq \sample$, there exists some parameter $\vec{\rho} \in \configs$ such that for all elements $x_i \in T$,
  $u_{\vec{\rho}}\left(x_i\right) \leq z_i$ and for all $x_i \not\in T$,
  $u_{\vec{\rho}}\left(x_i\right) > z_i$. We say that $\sample$ is {\em shattered} by $\{u_{\vec{\rho}}(\cdot)\}_{\vec{\rho}}$ if there exist $z_1, ..., z_m$ that witness its shattering.
  \item[(b)]  Let $\sample \subseteq \Pbi$ be the largest set that can be shattered by $\{u_{\vec{\rho}}(\cdot)\}_{\vec{\rho}}$. Then the pseudo-dimension of the class $\{u_{\vec{\rho}}(\cdot)\}_{\vec{\rho}}$ is $\pdim(\{u_{\vec{\rho}}(\cdot)\}_{\vec{\rho}}) = |\sample|$.
  \end{itemize}
\end{definition}

When $\{u_{\vec{\rho}}(\cdot)\}_{\vec{\rho}}$ is a set of binary valued functions, the notion of pseudo-dimension reduces to the notion of VC-dimension. 
%
\begin{theorem}~\label{thm:pdim}
Let $d_{\alg}$ be the pseudo-dimension of the family of utility functions $\{u_{\vec{\rho}}(\cdot)\}_{\vec{\rho}}$  induced by the class of algorithms $\mathcal{A}$ and the utility function $u(x, \vec{\rho})$; assume that the range of $u(x, \vec{\rho})$ is $[0,H]$.
For any $\epsilon > 0$, any $\delta \in (0,1)$ and any distribution $\dist$ over $\Pbi$,
$m = O \left(\frac{H^2}{\epsilon^2}\left( d_{\alg} + \ln\frac{1}{\delta}\right) \right)$ samples are sufficient to ensure that with probability  $1-\delta$
  over the draw of $m$ samples $\sample = \{x_1, \dots, x_m\} \sim \dist^m$, for
  all  $\vec{\rho} \in \configs$, the difference between the average utility of the algorithm $A_{\vec{\rho}}$  over the samples and its expected utiliy  is
  $ \leq  \epsilon$, i.e.: $$\left|\frac{1}{m}\sum_{i = 1}^m
  u_{\vec{\rho}}(x_i) - \E_{x \sim \dist}\left[u_{\vec{\rho}}(x)\right]\right| \leq
  \epsilon.$$
\end{theorem}

Theorem~\ref{thm:pdim} implies that to obtain sample complexity guarantees it is sufficient to bound the pseudo-dimension of the family
 $\{u_{\vec{\rho}}(x)\}_{\vec{\rho} \in \configs}$.
Interestingly, many of the proofs in the literature for doing this proceed by providing (either implicitly or explicitly) a structural result for the dual class of functions, $\{u_x(\vec{\rho})\}_{x \in \Pbi}$, where  $u_x(\vec{\rho}) = u(x,\vec{\rho})=u_{\vec{\rho}}(x)$.
We present below a simple, but powerful  lemma  of this form, which we will use throughout the chapter for the case that our parameter vector $\vec{\rho}$ is just a single real number; this lemma is used implicitly or explicitly in several papers~\citep{GR16,BNVW17,BDW18}.
\begin{lemma}\label{lem:piecewiseconstant}
Suppose that for every instance $x \in \Pbi$, the function $u_x(\rho): \R\rightarrow \R$ is piecewise constant with at most $N$ pieces.  Then the family $\{u_\rho(x)\}$ has pseudo-dimension $O(\log N)$.
\end{lemma}
\begin{proof}
Consider a problem instance $x \in \Pbi$.  Since the function $u_x(\rho)$ is piecewise constant with at most $N$ pieces, this means there are at most $N-1$ {\em critical points} $\rho_1^*, \rho_2^*, \ldots$ such that between any two consecutive critical points $\rho_i^*$ and $\rho_{i+1}^*$, the function $u_x(\rho)$ is constant.

Consider $m$ problem instances $x_1, \ldots, x_m$.  Taking the union of their critical points and sorting them, between any two consecutive of these critical points we have that {\em all} of the functions $u_{x_j}(\rho)$ are constant.  Since these critical points break up the real line into at most $(N-1)m+1 \leq Nm$ intervals, and all $u_{x_j}(\rho)$  are constant in each interval, this means that overall there are at most $Nm$ different $m$-tuples of values produced over all $\rho$.  Equivalently, the functions $u_\rho(x)$ produce at most $Nm$ different $m$-tuples of values on the $m$ inputs $x_1, \ldots, x_m$. However, in order to shatter the $m$ instances, we must have  $2^m$ different $m$-tuples of values produced.  Solving $Nm \geq 2^m$ shows that only sets of instances of size
$m=O(\log N)$ can be shattered.
\end{proof}


\subsection{Greedy Algorithms for Subset Selection Problems}
\label{se:knapsack}
In this section, we discuss  infinite parametrized families of greedy algorithms for subset selection problems introduced and analyzed in~\citet{GR16}. We start by discussing a specific family of algorithms for the canonical knapsack problem and then present a general result applicable to other problems including maximum weight independent set.

\paragraph{Knapsack} For the knapsack problem, let $\alg_{knapsack} = \{A_\rho\}$ be  the family of greedy algorithms described earlier. For this family $\configs=\R_{\geq 0}$, and for $\rho \in \configs$, for an instance $x$ where $v_i$ and $s_i$ are the value and
size of item $i$, the algorithm
$A_{\rho}$ adds the items to the knapsack in
decreasing order of $v_i/s_i^\rho$ subject to the capacity constraint. The
 utility function  $u(x, \rho)$ is defined as the value of the items chosen by
the greedy algorithm with parameter $\rho$ on input $x$.
 %
 %
We can show that the class $\alg_{knapsack}$ is not too complex, in the sense that its pseudo-dimension is small.

\begin{theorem}
\label{knapsack}
The family of utility functions $\{u_\rho(x)\}$ corresponding to $\alg_{knapsack}$ has pseudo
-dimension $O(\log n)$, where $n$ is the maximum number of items in an instance.
\end{theorem}

\begin{proof}
We first show that each function $u_x(\rho)$ is piecewise constant with at most $n^2$ pieces, and then apply Lemma \ref{lem:piecewiseconstant}.

To show the first part, fix some instance $x$.  Now, suppose  algorithm $A_{\rho_1}$ produces a different solution on $x$ than $A_{\rho_2}$ does for $\rho_1 < \rho_2$.  We argue there must exist some critical value $\crit \in [\rho_1,\rho_2]$ and some pair of items $i,j \in x$ such that $v_i/s_i^{\crit} = v_j/s_j^{\crit}$.  The reason is that if $A_{\rho_1}$  and $A_{\rho_2}$ produce different solutions on $x$, they must at some point make different decisions about which item to add to the knapsack.  Consider the first point where they differ: say that $A_{\rho_1}$  adds item $i$ to the knapsack and $A_{\rho_2}$ adds item $j$.  Then it must be the case that $v_i/s_i^{\rho_1} - v_j/s_j^{\rho_1}  \geq 0$ but $ v_i/s_i^{\rho_2} - v_j/s_j^{\rho_2} \leq 0$.    Since the function $f(\rho) = v_i/s_i^{\rho} - v_j/s_j^{\rho}$ is continuous,  there must exist some value $\crit \in [\rho_1,\rho_2]$ such that $v_i/s_i^{\crit} - v_j/s_j^{\crit} = 0$ as desired.

Now, for any given pair of items $i,j$, there is at most one value of $\rho\geq 0$ such that $v_i/s_i^\rho = v_j/s_j^\rho$; in particular, it is $\rho = \log(v_i/v_j)/\log(s_i/s_j)$.\footnote{Except for the special case that $s_i=s_j$ and $v_i=v_j$, but in that case the order the items are considered in is fixed by the tie breaking rule, so we can ignore any such pair.}  This means there are at most ${n \choose 2}$ critical values $\crit$ such that $v_i/s_i^{\crit} = v_j/s_j^{\crit}$ for some pair of items $i,j \in x$.
 By the argument above,
 all values of $\rho$ in the interval between any two consecutive
 critical values must produce the same behavior on the instance $x$.
This means there are at most ${n \choose 2}+1 \leq n^2$ intervals  such that all values of $\rho$ inside the same interval result in the exact same solution by algorithm $A_{\rho}$.

Now, we simply apply Lemma \ref{lem:piecewiseconstant} with $N=n^2$.
\end{proof}

\paragraph{ Maximum weighted independent set}
Another canonical subset selection problem is the maximum weighted independent set problem (MWIS).  An instance $x$ is a graph
with a weight $w\left(v\right) \in \R_{\geq 0}$ for each vertex $v$. The goal is to
find a set of mutually non-adjacent vertices with maximum total weight.
\cite{GR17} analyze a family $\alg_{MWIS}$ of
greedy heuristics that at each step selects the vertex maximizing
$w\left(v\right)/\left(1 + \deg\left(v\right)\right)^{\rho}$, where
$\rho \in \configs = [0, B]$ for some $B \in \R$, and then removes $v$ and its neighbors from the graph. Using a similar argument as in Theorem~\ref{knapsack} we can show that the family of utility functions $\{u_\rho(x)\}$ corresponding to $\alg_{MWIS}$ has pseudo-dimension $O(\log n)$, where  $n$ is the maximum number of vertices in an instance.

\paragraph{A general analysis for greedy heuristics}
We now more generally consider problems where the input is a set of
$n$ objects with various attributes, and the feasible solutions consist of assignments
of the objects to a finite set $Y$, subject to feasibility constraints. The attributes of
an object are represented as an element $\xi$ of an abstract set. For example, in the
Knapsack problem $\xi$ encodes the value and size of an object; in the MWIS problem,
$\xi$ encodes the weight and (original or residual) degree of a vertex. In the Knapsack
and MWIS problems, $Y= \{0, 1\}$, indicating whether or not a given object is selected.

~\cite{GR17} provide pseudo-dimension bounds for general greedy heuristics of the following form:

\begin{itemize}
\item[] While there remain unassigned objects,
\begin{enumerate}
\item[(a)] Use a scoring rule $\sigma$ (a function from attributes to $\R$) to compute a score $\sigma(\xi_i)$ for each unassigned object $i$, as a function of its current attributes $\xi_i$.
\item[(b)] For the unassigned object $i$ with the highest score, use an assignment
rule to assign $i$ a value from $Y$ and, if necessary, update the attributes
of the other unassigned objects. Assume that ties are
always resolved lexicographically.
\end{enumerate}
\end{itemize}

Assignment rules that do not modify objects' attributes yield nonadaptive greedy heuristics, which use only the original attributes of each object (like $v_i$ or $v_i/s_i$
in the Knapsack problem, for instance).
Assignment rules that modify object attributes yield adaptive greedy heuristics,
such as the adaptive MWIS heuristic described above. In a {\em single-parameter} family of scoring rules, there is a scoring rule of the form
$\sigma(\rho, \xi)$ for each parameter value $\rho$ in some interval $I \subseteq \reals$. Moreover, $\sigma$ is assumed
to be continuous in $\rho$ for each fixed value of $\xi$. Natural examples include Knapsack
scoring rules of the form $v_i/s_i^\rho$
and MWIS scoring rules of the form $w(v)/(1+deg(v))^\rho$
for $\rho \in [0, 1]$ or $\rho \in [0,\infty)$.

A single-parameter family of scoring rules is {\em $\kappa$-crossing} if, for
each distinct pair of attributes $\xi'$, $\xi''$, there are at most $\kappa$ values of $\rho$ for which
$\sigma(\rho, \xi') = \sigma(\rho, \xi'')$.
For example, all of the scoring rules mentioned above are $1$-crossing rules.

For an example assignment rule, in the Knapsack and MWIS problems, the rule
simply assigns $i$ to $1$ if it is feasible to do so, and to $0$ otherwise. In the adaptive greedy heuristic for the MWIS problem, whenever the assignment rule assigns $1$ to a vertex $v$, it updates the residual degrees of other unassigned
vertices (two hops away) accordingly. Say that an assignment rule is {\em $\beta$-bounded} if every object $i$ is guaranteed to take
on at most $\beta$ distinct attribute values. For example, an assignment rule that never
modifies an object's attributes is $1$-bounded. The assignment rule in the adaptive
MWIS algorithm is n-bounded, since it only modifies the degree of a vertex (which
lies in $\{0, 1, 2, \ldots, n-1 \}$).
Coupling a single-parameter family of $\kappa$-crossing scoring rules with a fixed $\beta$-bounded assignment rule yields a $(\kappa,\beta)$-single-parameter family of greedy heuristics.
The knapsack greedy heuristic is a $(1, 1)$-single-parameter family and the adaptive MWIS heuristic is a $(1, n)$-single-parameter family.

\begin{theorem}
Let $\alg_{greedy}$ be a $(\kappa,\beta)$  single parameter family of greedy heuristics and let $\{u_\rho(x)\}$ be its corresponding family of utility functions.
The pseudo-dimension of $\{u_\rho(x)\}$ is $O(\log(\kappa \beta n))$, where $n$ is the number of objects.
\end{theorem}

\begin{proof}
Fix an instance $x$, and consider the behavior of the algorithm as we vary $\rho$.   Because there are $n$ items and the assignment rule is $\beta$-bounded, there are a total of at most $n\beta$ distinct attribute values possible over all choices of $\rho$.  For any two such attribute values $\xi',\xi''$, we know by the $\kappa$-crossing assumption there are at most $\kappa$ distinct {\em critical values} $\crit$ such that  $\sigma(\crit,\xi') = \sigma(\crit,\xi'')$.  Thus, there are at most $(n\beta)^2\kappa$ distinct critical values total.  Now, between any two consecutive critical values, the algorithm must behave identically for all  $\rho$ in that interval.  In particular, if $\rho_1$ and $\rho_2$ behave differently on $x$, there must exist two attribute values $\xi', \xi''$ such that one has higher score under $\rho_1$ but the other has higher score under $\rho_2$, and by continuity of $\sigma$ this means $\rho_1$ and $\rho_2$ must be separated by a critical value.  Since the algorithm behaves identically in each interval and there are at most  $(n\beta)^2 \kappa +1$ intervals, this means that  that each function $u_x(\rho)$ is piecewise constant with at most $(n\beta)^2 \kappa +1$ pieces.  The theorem then follows from Lemma \ref{lem:piecewiseconstant}.
\end{proof}

\subsection{Clustering problems}\label{sec:clusteringdistributional}
In this section we discuss how a data-driven approach can help overcome impossibility results for clustering problems.
Clustering is one of the most basic problems in  data science; given a large set of complex data (e.g., images or news articles) the goal is to organize it into groups of similar items.
Despite significant efforts from  different communities, it remains a major challenge.
Traditional approaches  have focused on the ``one shot'' setting, where the goal is to cluster a single potentially worst-case dataset.
Unfortunately, there are major impossibility results for such scenarios; first, in most applications it is not clear  what objective function to use in order to recover a good clustering for the given data set; second, even in cases where the objective can be naturally specified, optimally solving the underlying combinatorial clustering problem is typically intractable.
 One approach to circumvent hardness of worst case instances 
 is to posit specific stability assumption about the input instances, and to design efficient algorithms with good performance on such instances.
 Another approach that is particularly suited  for settings  (including text and image categorization) where we  have to solve  many clustering problems arising in a given application domain, is to select  a good clustering algorithm in a data-driven way.
 In particular, given a series of clustering instances to be solved from the same domain,  we  learn a  good
parameter setting for a clustering algorithm (from a large potentially infinite set of clustering algorithms)
that performs well on instances coming from that domain. We can then use   the general framework discussed in Section~\ref{se:dislearn} to provide
guarantees for this approach. We discuss below such guarantees  for several parametric families of clustering procedures widely used in practice.

\paragraph{Problem Setup}
The results we  present apply both to objective based clustering (e.g., k-means and k-median) and to an unsupervised learning formulation of the problem.   In both cases the input to a clustering problem is a point set $V$ of $n$ points, a desired number of clusters $k \in \{1, \ldots, n\}$, and a metric $d$ (such as Euclidean distance in $\mathbb{R}^d$) specifying the distance between any two points; throughout the rest of this section we denote by
$d(i,j)$  the distance between points $i$ and $j$.

For objective based clustering, the goal is to output
a partition $\mathcal{C}=\{C_1,\dots,C_k\}$ of $V$ that optimizes a specific objective function.
For example, in the k-means clustering objective the goal is to output
a partition $\mathcal{C}=\{C_1,\dots,C_k\}$
and a center $c_i$ for each $C_i$
in order to minimize  the  sum of the squared distances between every point and its nearest center, i.e.
$\text{cost}(\mathcal{C})=\left(\sum_i\sum_{v\in C_i}d(v,c_i)^2\right)$, while in the $k$-median objective the goal is to
minimize the sum of distances to the centers rather than the squared distances, i.e. $\text{cost}(\mathcal{C})=\left(\sum_i\sum_{v\in C_i}d(v,c_i)\right)$.
Unfortunately, finding the clustering that minimizes  these objectives (and other classic ones such as k-center and min-sum) is NP-hard, so using a data-driven approach can  help in identifying solutions with good objective values for specific domains.

In the unsupervised learning or ``matching the ground-truth clustering'' approach, we assume that for each instance  $V$ of $n$ points, in addition to the distance metric $d$,  there is a ground-truth partition of the input points $\mathcal{C^*}=\{C_1^*,\dots,C_k^*\}$. The goal is to output a partition $\mathcal{C}=\{C_1,\dots,C_k\}$ in order to minimize some loss function relative to the ground-truth; e.g., a common loss function 
is the fraction of points that would have to be reassigned in $\mathcal{C}$ to make it match $\mathcal{C^*}$ up to re-indexing of the clusters, or equivalently $\min_{\sigma} \frac{1}{n} \sum_{i=1}^k |C_i \setminus C^*_{\sigma(i)}|,$
where the minimum is taken over all bijections $\sigma: \{1,\ldots,k\} \rightarrow \{1,\ldots,k\}$. For the data-driven approach we assume that the ground-truth is known for the training instances, but  it is unknown and what we want to predict for the test instances.

\paragraph{Linkage-based families}
 In the following we discuss families of two stage clustering algorithms, that in the first stage use a linkage procedure to organize data into a hierarchical clustering and then in a second stage use a fixed (computationally efficient) procedure to extract a pruning from this hierarchy. Such techniques are prevalent in practice and from a theoretical point of view, they are known to perform nearly optimally in  settings where the data is well-clusterable, in particular perturbation resilient and approximation stable. 

The linkage procedure in the first step takes as input a clustering instance $x$ (a set  $V$ of $n$ points  and metric $d$ specifying the distance between any pair
  of the base points)
and outputs a
cluster tree, by repeatedly merging the two closest clusters.
In particular, starting with the base distance $d$, we first define a distance measure $D(A,B)$ between any two subsets $A$ and $B$  of $\{1,\ldots,n\}$, that is used to greedily link the data into a binary cluster tree.  The leaves of the tree are the individual data points, while the root
node corresponds to the entire dataset.
The algorithm starts with each
point belonging to its own cluster. Then, it repeatedly merges the closest pair
of clusters according to  distance  $D$. 
When there is only a single cluster remaining, the algorithm
outputs the constructed cluster tree.  Different definitions for  $D$ lead to different hierarchical procedures. For example, the classic linkage procedures single,  complete, and average linkage define $D$ as  $D(A,B)= \dmin(A,B)= \min_{a \in A, b \in B} d(a,b)$,  $D(A,B) = \dmax(A,B)= \max_{a \in A, b
\in B} d(a,b)$, and $D(A,B)  =  \frac{1}{|A||B|}\sum_{u \in A, v \in B} d(u, v)$,
respectively.

The procedure in the second step can be as simple as just ``undoing'' the last $k-1$ merges from the first step  or a dynamic programming subroutine over the hierarchy from the first step to extract a clustering of highest score based on some measurable objective such as $k$-means or $k$-median cost.
The final  quality or utility (measured by the function $u_\rho(x)$ on clustering instance $x$) of the solution produced by the algorithm is for the objective based approach measured by the  given objective function (e.g. $k$-means or $k$-median objective) or the loss with respect to the ground truth in the unsupervised learning formulation.


We analyze below the pseudo-dimension of two parametric families of algorithms of this form (from~\cite{BNVW17}).
  Both of these families use a parametrized linkage procedure in the first step, and
the cluster tree produced  is then fed into a fixed second-stage procedure to produce a $k$-clustering.
 The first family $\algfamily^{scl}$ uses a parametrized family of linkage  algorithms with a single parameter $\rho \in \configs=[0,1]$  that helps interpolate linearly between the classic single and complete linkage procedures. For  $\rho \in \configs$ the algorithm  $A_\rho \in \algfamily^{scl}$ defines the distance between two sets $A$ and $B$ as $$D_{\rho}^{scl}(A, B) = (1-\rho) \dmin(A,
B) + \rho \dmax(A,B).$$  Note that $\rho = 0$ and
$\rho = 1$ recover single and complete linkage, respectively.

The second family $\algfamily^{exp}$ uses a parametrized family of linkage  algorithms with a single parameter $\rho \in \configs=\reals$  that helps interpolate not only between single and complete linkage but also includes average linkage as well. For  $\rho \in \configs$ the algorithm
$A_\rho \in \algfamily^{exp}$ defines the distance between two sets $A$ and $B$ as
$$
D_{\rho}^{exp}(A, B)  =  \left(\frac{1}{|A||B|}\sum_{u \in A, v \in B} \left(d(u, v)\right)^{\rho}\right)^{1/\rho}.
$$
Note that $\rho=0$ recovers average linkage, $\rho \rightarrow \infty$  recovers complete linkage, and $\rho \rightarrow -\infty$ recovers single linkage.
~\cite{BNVW17} prove that the family of functions $\{u_\rho(x)\}$ corresponding to the family $\algfamily^{scl}$ is not too complex, in the sense that it has pseudo-dimension $\Theta(\log n)$, where $n$ is an upper bound on the number of data points in a clustering instance. Similarly, the family of functions $\{u_\rho(x)\}$ corresponding to the family $\algfamily^{exp}$ has pseudo-dimension $\Theta(n)$. We sketch the upper bounds below.

 We start by analyzing the family $D_{\rho}^{scl}$-linkage, for which we can prove the following structural result.

\begin{lemma}\label{lem:single}
Let $x$ be a clustering instance.  We can partition $\configs$ into at most $n^8$ intervals such that all values of $\rho$ inside the same interval result in the exact same solution produced by the $D_{\rho}^{scl}$-linkage algorithm.
\end{lemma}

\begin{proof}
First, for any pair of candidate cluster merges $(C_1, C_2)$ and $(C'_1, C'_2)$, where
$C_1$, $C_2$, $C'_1$ and $C'_2$ are clusters, there is at most one critical parameter value
$\crit$ such that $D_{\rho}^{scl}(C_1, C_2) = D_{\rho}^{scl}(C'_1, C'_2)$ only when $\rho = \crit$.  In particular, $\crit = \deltamin / (\deltamin - \deltamax)$, where $\deltamin = \dmin(C'_1, C'_2) - \dmin(C_1, C_2)$ and $\deltamax = \dmax(C'_1, C'_2) - \dmax(C_1, C_2)$.  For clarity, we will call this value $\crit(C_1, C_2, C'_1, C'_2)$.

Next, the total number of distinct critical values $\crit$ ranging over all possible $4$-tuples of clusters $C_1, C_2, C'_1, C'_2$ is at most $n^8$.  The reason is that for any given  clusters $C_1, C_2, C'_1, C'_2$ there exist 8 points (not necessarily distinct) corresponding to the closest pair between $C_1$ and $C_2$, the closest pair between $C_1'$ and $C_2'$, the farthest pair between $C_1$ and $C_2$, and the farthest pair between $C_1'$ and $C_2'$, whose distances completely define $\crit(C_1, C_2, C'_1, C'_2)$.  Since there are at most $n^8$ possible 8-tuples of such points, this means there are at most $n^8$ distinct critical values.

Between any two consecutive critical values $\crit$, all $D_{\rho}^{scl}$-linkage algorithms give the same ordering on all possible merges. This is because for any  $C_1, C_2, C'_1, C'_2$, the function $f(\rho) = D_{\rho}^{scl}(C_1,C_2) - D_{\rho}^{scl}(C_1',C_2')$ is continuous, and therefore must have a zero (creating a critical value) if it switches sign.   So, there are at most $n^8$ intervals such that all values of $\rho$ inside the same interval result in the exact same merges, and therefore the same solution produced by the $D_{\rho}^{scl}$-linkage algorithm.
\end{proof}

Lemma~\ref{lem:single} and Lemma~\ref{lem:piecewiseconstant} imply the following:
\begin{theorem}
The family of functions $\{u_\rho(x)\}$ corresponding to the family $\algfamily^{scl}$-linkage has pseudo-dimension $O(\log n)$.
\end{theorem}

\begin{theorem}
The family of functions $\{u_\rho(x)\}$ corresponding to the family $\algfamily^{exp}$-linkage has pseudo-dimension $O(n)$.
\end{theorem}
\begin{proof} (Sketch)
As in the proof of Lemma \ref{lem:single}, we fix an instance $x$ and bound the number of  intervals such that all values of $\rho$ inside the same interval result in the exact same solution produced by the algorithm.

Fixing an instance $x$, consider two pairs of sets $A,B$ and $X,Y$ that could be potentially merged.  Now, the decision to merge one pair before the other is determined by the sign of the expression
$$ \frac{1}{|A||B|}\sum_{p \in A, q \in B} (d(p,q))^\rho - \frac{1}{|X||Y|}\sum_{x \in X, y \in Y} (d(x,y))^\rho.$$ First note that this expression has $O(n^2)$ terms, and
by a consequence of Rolle's Theorem, it has $O(n^2)$ roots. Therefore, as we iterate over the $O\left(\left(3^{n}\right)^2\right)$ possible pairs $(A,B)$ and $(X,Y)$, we can determine $O\left(3^{2n}\right)$ unique expressions each with $O(n^2)$ values of $\rho$ at which the corresponding decision flips. Thus, by continuity of the associated functions, we can divide $\R$ into at most $O\left(n^2 3^{2n}\right)$ intervals over each of which the output of the algorithm on input $x$ is fixed.

Finally, we apply Lemma \ref{lem:piecewiseconstant} using the fact that each function $u_x(\rho)$ is piecewise constant with at most $2^{O(n)}$ pieces.
\end{proof}

Interestingly, these  families of clustering algorithms are also known to have strong analytical properties for  stable instances.
One such condition, called perturbation-resilience, asks that even if distances between data points are perturbed by up to some factor $\beta$,  the clustering that optimizes a given objective (such as $k$-means or $k$-median) does not change.  If this condition is satisfied for $\beta \geq 2$, it is known that one can find the optimal clustering efficiently, in fact via a linkage algorithm followed by dynamic programming, 
further motivating that algorithm family.
However, one drawback of all these results is that if the condition does not hold, the guarantees do not apply.  Here, we aim to provide guarantees on optimality within an algorithm family that hold regardless of clusterability assumptions, but with the additional property that if typical instances are indeed well-clusterable (e.g., they satisfy perturbation-resilience or some related condition), then the optimal algorithm in the family is optimal overall.  This way, we can produce guarantees that simultaneously are meaningful in the general case and can take advantage of settings where the data is particularly well behaved.

\paragraph{Parametrized Lloyd's methods}
The Lloyd's method is another popular technique in practice. The procedure starts with $k$ initial centers and iteratively makes incremental improvements until a local optimum is reached. One of the most crucial decisions an algorithm designer must make when using such an algorithm is the initial seeding step, i.e.,  how the algorithm chooses the $k$ initial centers.
~\cite{BDW18} consider an infinite family of algorithms generalizing
the popular $k$-means++ approach \citep{arthur2007k}, with a parameter $\alpha$ that controls the seeding process.
In the seeding phase, each point $v$ is sampled with probability proportional to
$d_{\text{min}}(v,C)^\alpha$, where $C$ is the set of centers chosen so far and $d_{\text{min}}(v,C) = \min_{c \in C}d(v,c)$.
Then Lloyd's method is used to converge to a local minimum or is cut off at some given time bound.  By ranging over $\alpha\in [0,\infty)\cup\{\infty\}$, we obtain an infinite family of algorithms which we call $\alpha$-Lloyds++.
This allows a spectrum between random seeding ($\alpha=0$),
and farthest-first traversal ($\alpha=\infty$), with  $\alpha=2$ corresponding to $k$-means++.
 What is different about this algorithm family compared to those studied earlier in this chapter is that because the algorithm is randomized, for this problem the expected cost as a function of $\alpha$ is Lipschitz.  In particular, one can prove a Lipschitz constant of $O(nkH\log R)$, where $R$ is the ratio of maximum to minimum pairwise distance, and $H$ is an upper bound on the $k$-means cost of any clustering. As a consequence, one can discretize values of $\alpha$ into a fine grid, and then try $N=O(\alpha_h nkH (\log R)/\epsilon)$ values of $\alpha$ on a sample of size $O((H/\epsilon)^2\log N)$ and pick the best, where $\alpha_h$ is the largest value of $\alpha$ one wishes to consider.  However, by pushing the randomness of the algorithm into the problem instance (augmenting each problem instance with a random string and viewing the algorithm as a deterministic function of the instance and random string), one can view $u_x(\alpha)$ as a piecewise-constant function with a number of pieces that {\em in expectation} is only $O(nk(\log n)\log (\alpha_h \log R))$.  This allows for many fewer values of $\alpha$ to be tried, making this approach  more practical. In fact, \cite{BDW18} implement this approach and demonstrate it on several interesting datasets.

\subsection{Other Applications and Generic Results}
\label{generic}

\paragraph{Partitioning problems via IQPs}
~\cite{BNVW17} study data-driven algorithm design for problems that can be written as integer quadratic programs (IQPs)
for families of algorithms that involve semidefinite programming (SDP) relaxations followed by parametrized rounding schemes.
The class of IQP problems they consider is described as follows.  An instance $x$
is specified by a matrix $A \in \R^{n \times n}$, and the goal is to solve (at least approximately) the  optimization problem $\max_{\vec{z} \in \{\pm 1\}^n} \vec{z}^\top A \vec{z}$.
This is of interest since many classic NP-hard
problems can be formulated as IQPs, including max-cut,
 max-2SAT, and correlation clustering.
For example, the classic max-cut problem can be written as an IQP of this form.  Recall that given a graph $G$ on $n$ nodes with edge weights $w_{ij}$, the max-cut problem is to find a partition of the vertices into two sides to maximize the total sum of edge weights crossing the partition.  This can be written as solving for $\vec{z} \in \{\pm 1\}^n$ to maximize $\sum_{(i,j)\in E} w_{ij}\left(\frac{1-z_i z_j}{2}\right)$, where $z_i$ represents which side vertex $i$ is assigned to.
This objective can be formulated as $\max_{\vec{z} \in \{\pm 1\}^n} \vec{z}^\top A \vec{z}$ for $a_{ij}=-w_{ij}/2$ for $(i,j)\in E$ and $a_{ij}=0$ for $(i,j)\not\in E$.

The  family of algorithms $A_{\rho}^{round}$ that~\citet{BNVW17} analyze is parametrized by a one dimensional set $\configs = \reals$, and for any $\rho \in \configs$ the algorithm $A_{\rho}$ operates as follows. In the first stage it solves the SDP relaxation
$\sum_{i,j \in [n]} a_{ij} \langle \vec{u}_i, \vec{u}_j \rangle$ subject to the
constraint that $\norm{\vec{u}_i} = 1$ for $i \in \{1, 2, \ldots, n\}$. 
In the second stage it rounds the vectors $\vec{u}_i$ to $\left\{\pm 1\right\}$ by
sampling a standard Gaussian $\vec{Z} \sim \mathcal{N}_n$ and setting $z_i = 1$ with
probability $1/2 + \phi_{\rho}\left(\langle \vec{u}_i, \vec{Z}
\rangle\right)/2$ and $-1$ otherwise, where
$\phi_{\rho}(y) = y/\rho$ for $-\rho \leq y \leq \rho$, $\phi_{\rho}(y)=-1$ for $y < -\rho$, and $\phi_{\rho}(y)=1$ for $y>\rho$.
In other words, if $|\langle \vec{u}_i, \vec{Z}\rangle|> \rho$ then $\vec{u}_i$ is rounded based on the sign of the dot-product, else it is rounded probabilistically using a linear scale.
The utility function $u_{\rho}(x)$ for algorithm $A_{\rho}$ maps the algorithm parameter  $\rho$
to the expected objective value obtained on the instance $x$. Note that by design the algorithms $A_{\rho}^{round}$ are polynomial time algorithms.

Note that when $\rho=0$, this algorithm corresponds to the classic Goemans-Williamson max-cut
algorithm.  It is known that nonzero values of $\rho$ can
 outperform the classic algorithm on graphs for which the max cut does
 not constitute a large fraction of the edges.

\begin{theorem}
Let $\{u_\rho(x)\}$ be the corresponding family of utility functions for the family of algorithms $A_{\rho}^{round}$.
The pseudo-dimension of $\{u_\rho(x)\}$ is $O(\log(n))$, where $n$ is the maximum number of variables in an instance.
\end{theorem}

At a high level, the proof idea is to analyze a related utility function where we imagine that the
Gaussians $\vec{Z}$ are sampled ahead of time and included as part of the
problem instance; in other words we augment the instance to obtain a new instance $\tilde{x}=(x,\vec{Z})$. One can then prove that the utility is
$u_{\rho}(\tilde{x}) = \sum_{i = 1}^n a_{ii}^2 + \sum_{i \not= j} a_{ij} \phi_s(v_i) \phi_s(v_j)$, where $v_i =
\langle \vec{u_i}, \vec{Z} \rangle$. Using this form, it is easy to show that this objective function  value is piecewise quadratic in $1/\rho$ with $n$ boundaries. The result then follows from a generalization of Lemma~\ref{lem:piecewiseconstant}.

\paragraph{Learning to branch}
So far we considered families of polynomial time algorithms and scored them based on solution quality (e.g., clustering quality or objective value).
 In general, one could also score algorithms based on other important measures of performance. For example,~\cite{BDSV18} consider parametrized branch-and-bound techniques for learning how to branch when solving mixed integer programs (MIPs) in the distributional learning setting, and score a parameter setting  based on the tree size on a given instance (which roughly corresponds to running time). ~\cite{BDSV18} show that the corresponding dual functions are piecewise constant, and then the sample complexity results follow from a high-dimensional generalization of Lemma~\ref{lem:piecewiseconstant}.
 ~\cite{BDSV18}  also show experimentally  that different parameter settings of these families of algorithms can result in branch and bound trees of vastly different sizes, for different combinatorial problems (including  winner determination in combinatorial auctions, k-means clustering, and agnostic learning of linear separators). They also show that the optimal parameter is highly distribution-dependent: using a parameter optimized on the wrong distribution can lead to a dramatic tree size blowup, implying that learning to branch is both practical and hugely beneficial.

\paragraph{General bounds via structure of dual functions}
~\citet{BDDKSV19} present a general sample complexity result applicable to  algorithm configuration problems for which the dual functions are piece-wise structured. The key innovation is to provide an elegant and widely applicable abstraction that simultaneously covers all the types of dual structures appearing in the algorithm families mentioned so far -- this includes those in Sections~\ref{se:knapsack} and ~\ref{sec:clusteringdistributional} (where $u_x(\rho)$ are piecewise-constant with a limited number of pieces as in Lemma~\ref{lem:piecewiseconstant}), and the dual functions appearing in the context of learning to branch mentioned above, as well as  revenue maximization in multi-item multi-bidder settings. ~\citet{BDDKSV19}  show that this theorem recovers all the prior results and they also show new applications including to dynamic programming techniques for important problems in computational biology, e.g.,  sequence alignment and protein folding.

Recall that  $\params$ denotes the space of parameter vectors $\vec{\rho}$ (e.g., if $\vec{\rho}$ consists of $d$ real-valued parameters, then $\params = \R^d$).
Let ${\cal F}$ denote a family of {\em boundary functions} such as linear separators or quadratic separators that each partition $\params$ into two pieces, and let ${\cal G}$ denote a family of {\em simple utility functions} such as constant functions or linear functions over $\params$.  \citet{BDDKSV19}  show the following.  Suppose that for each dual function $u_x(\vec{\rho})$, there are a limited number of boundary functions $f_1, \ldots, f_N \in {\cal F}$ such that within each region\footnote{Formally, each $f_i$ is a function from $\params$ to $\{0,1\}$, and a {\em region} is a nonempty set of $\vec{\rho}$ that are all labeled the same way by each $f_i$.}  defined by these functions, $u_x(\vec{\rho})$ behaves as some function from ${\cal G}$.  Then, the pseudo-dimension of the primal family $\{u_{\vec{\rho}}(x)\}$ can be bounded as a function of $N$,  the  VC-dimension of the dual class ${\cal F}^*$ to ${\cal F}$, and the pseudo-dimension of the dual class ${\cal G}^*$ to ${\cal G}$.\footnote{${\cal F}^*$ is defined as follows: for each $\vec{\rho} \in \params$ define the function  $\vec{\rho}(f) = f(\vec{\rho})$ for all $f \in {\cal F}$. ${\cal G}^*$ is defined similarly. }

\section{Data-driven Algorithm Design via Online Learning}
\label{se:online}
We now consider an online formulation for algorithm design where we do not assume that the instances of the given algorithmic problem are i.i.d. and presented all at once; instead, they could arrive online, in an arbitrary order, in which case what we can aim for is to compete with the best fixed algorithm in hindsight~\citep{BDV18,GR17,CAK17}, also known as no-regret learning.   Since the utility functions appearing in algorithm selection settings often exhibit sharp discontinuities, achieving no-regret  is impossible in the worst case over the input sequence of instances.

We discuss a niceness  condition on the sequence of utility functions introduced in \cite{BDV18}, called dispersion, that is sufficient for the existence of online algorithms that guarantee no regret.

\paragraph{Problem Formulation}
 On each round $t$ the learner chooses an algorithm from the family specified by the  parameter vector $\vec{\rho}_t$ and receives a new instance of the problem $x_t$; this induces the utility function $u_{x_t}(\vec{\rho})$ that measures the performance of each algorithm in the family for the given instance, and the utility of the learner at time $t$ is $u_{x_t}(\vec{\rho}_t)$. The case where the learner observes the entire utility function $u_{x_t}(\vec{\rho})$ or can evaluate it at points of its own choice is called the full information setting;  the case where it only observes the scalar $u_{x_t}(\vec{\rho}_t)$ is called the bandit setting.
 The goal is to select algorithms so that the cumulative performance of the learner is nearly as good as the best algorithm in hindsight for that sequence of problems. Formally, the goal is to minimize expected regret
$$\E[\max_{\vec{\rho} \in \configs} \sum u_{x_t}(\vec{\rho}) - u_{x_t}(\vec{\rho}_t)],$$
 where the expectation is  over the randomness in the learner's choices or over the randomness in the utility functions.
We aim to obtain expected regret that is sublinear in $T$, since in that case  the per-round average performance of the algorithm is approaching that of the best parameter in hindsight -- this is commonly referred as achieving ``no regret" in the online learning literature.

 As we have seen in the previous sections the utility functions appearing in algorithm selection settings often exhibit sharp discontinuities, and it  is known that even for one-dimensional cases, achieving no-regret guarantees for learning functions with sharp discontinuities is impossible, in the worst case. In essence, the problem is that if $I$ is an interval of parameters that have all achieved maximum utility so far, an adversary can choose the next utility function to randomly give either the left or right half of $I$ a utility of 0 and the other half a maximum utility, causing any online algorithm to achieve only half  of the optimum in hindsight.
~\cite{GR17} show that this is the case for online algorithm selection for the maximum weighted independent set problem for the family of algorithms discussed in section~\ref{se:knapsack}.

We now describe a general condition on the sequence of utility functions introduced, called dispersion, that is provably sufficient to achieve no-regret, introduced in~\citet{BDV18}. Roughly speaking, a collection of utility  functions $u_{x_1}, \dots, u_{x_\numfunctions}$ is dispersed if no small region of the space contains discontinuities for many of these functions.
 Formally:

 \begin{definition}
  \label{def:dispersion}
  Let $u_{x_1}, \dots, u_{x_\numfunctions} : \configs \to [0,H]$ be a collection of utility functions where
  $u_{x_i}$ is piecewise Lipschitz over a partition $\partition_i$ of $\configs$. We
  say that $\partition_i$ splits a set $A$ if $A$ intersects with at least two
  sets in $\partition_i$. The collection of functions is
  \emph{$(w,k)$-dispersed} if every ball of radius $w$ is split by at most $k$
  of the partitions $\partition_1, \dots, \partition_\numfunctions$. More generally, the
  functions are \emph{$(w,k)$-dispersed at a maximizer} if there exists a
  point $\vec{\rho}^* \in \argmax_{\vec{\rho} \in \configs} \sum_{i = 1}^\numfunctions
  u_i(\vec{\rho})$ such that the ball $B(\vec{\rho}^*,w)$ is split by at most
  $k$ of the partitions $\partition_1, \dots, \partition_\numfunctions$.
\end{definition}

\begin{figure}[t]
\centering
\begin{subfigure}{.47\textwidth}
\includegraphics[width=\textwidth]{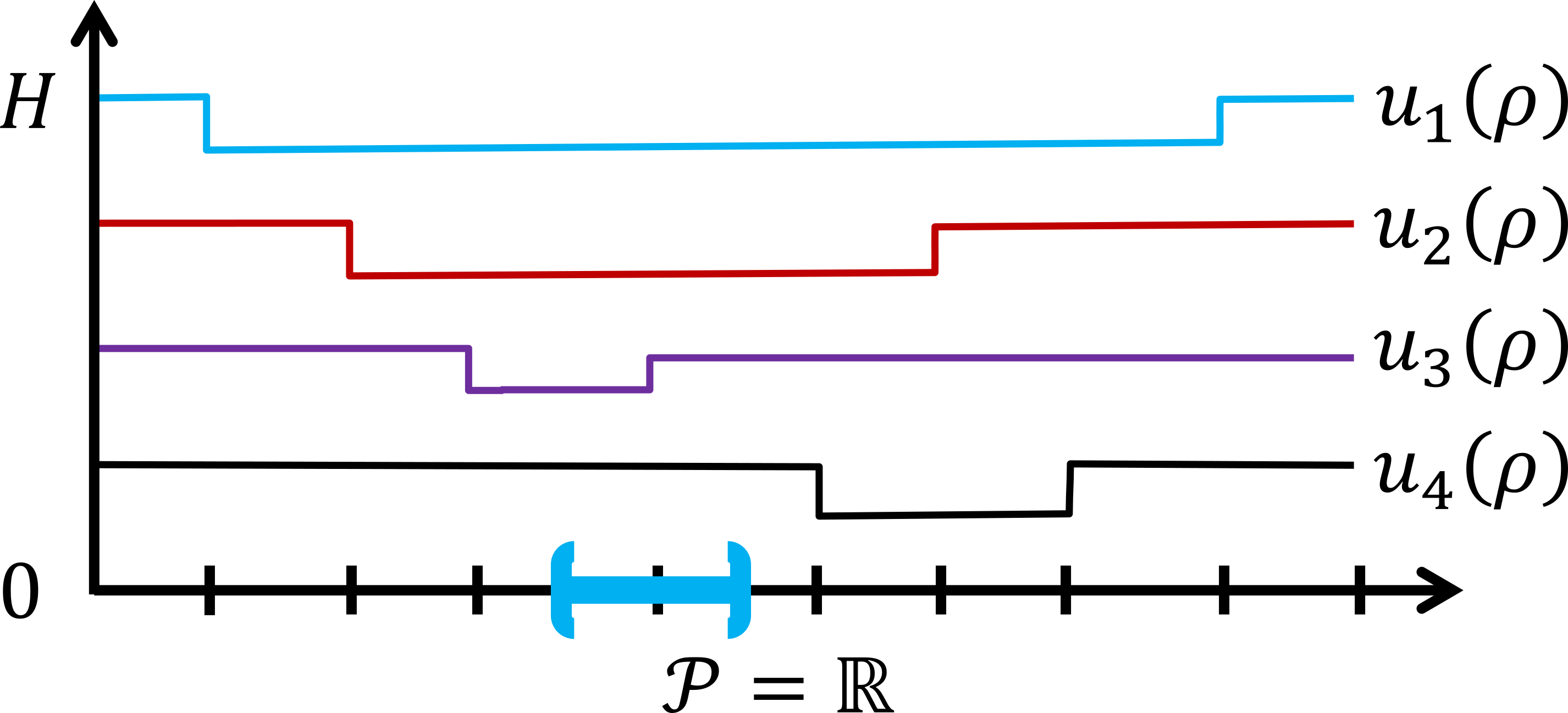}\centering
\caption{}\label{}
\end{subfigure}\qquad
\begin{subfigure}{.47\textwidth}
\includegraphics[width=\textwidth]{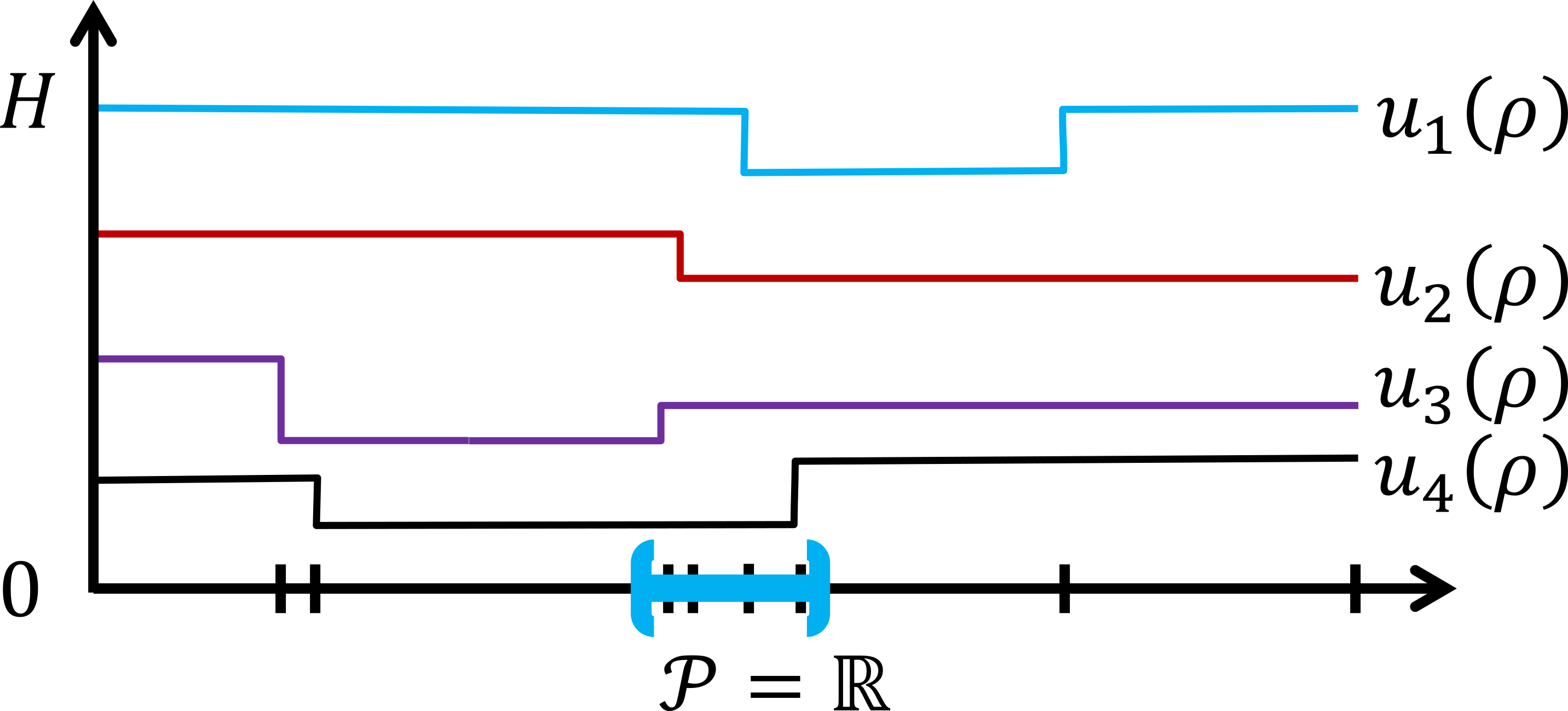}\centering
\caption{}\label{}
\end{subfigure}
\caption{The utility functions in figure (a) are dispersed because any small interval has discontinuities for only a few of them.  The utility functions in figure (b) are not dispersed because there is a small interval with many discontinuities.}
\label{}
\end{figure}


In many applications, Definition~\ref{def:dispersion}  holds with
 high probability for $w = \numfunctions^{\alpha - 1}$ and $k = \tilde O(\numfunctions^{\alpha})$ for some $1/2 \leq \alpha \leq 1$, ignoring  problem-specific multiplicands.

\paragraph{Continuous weighted majority}
In the full information setting, we can use a continuous version of the classic weighted majority algorithm~\citep{NicoloGabor06:PLG} to obtain no-regret learning in dispersed settings. In round $t$, the algorithm samples a vector
$\vec{\rho}_t$ from the distribution $$p_t(\vec{\rho}) \propto \exp(\lambda
\sum_{s = 1}^{t-1} u_s(\vec{\rho})).$$
The
 following bound holds for this algorithm~\citep{BDV18}.


\begin{theorem}\label{thm:1_d_online}
  Let $u_{x_1}, \dots, u_{x_{\numfunctions}} : \configs \to [0,H] $ be a sequence of utility functions
  corresponding to problem instances $x_1, \ldots x_\numfunctions$. Assume that these functions $u_{x_1}, \dots, u_{x_\numfunctions}$
  are piecewise
  $L$-Lipschitz functions and $(w,k)$-dispersed at the maximizer
  $\vec{\rho}^*$. Suppose $\configs \subset \reals^d$ is contained in
  a ball of radius $R$ and
  $B(\vec{\rho^*},w) \subset \configs$. The continuous weighted majority algorithm
  with $\lambda = \frac{1}{H}\sqrt{d\ln(R/w)/\numfunctions}$ has expected regret bounded by
  $O\left(H\left(\sqrt{\numfunctions d\log\frac{R}{w}} + k\right) + \numfunctions Lw\right).$
\end{theorem}

When $w = 1/\sqrt{\numfunctions}$ and $k = \tilde O(\sqrt{\numfunctions})$, this gives
 regret $\tilde O(\sqrt{\numfunctions}(H\sqrt{d}+L))$.

\begin{proof}[Proof sketch] Let $U_t$ be the function $\sum_{i = 1}^{t-1} u_{x_i}(\cdot)$; let $W_t$ be the normalizing constant at round $t$, that is $W_t = \int_\configs \exp(\lambda U_t(\vec{\rho})) \, d\vec{\rho}$.

The proof follows by providing upper and lower bounds on  $W_{\numfunctions+1}/W_1$.
The upper bound on $W_{\numfunctions+1}/W_1$ in terms of the learner's expected payout follows as in the classic weighted majority algorithm, yielding:
$$\frac{W_{T+1}}{W_1} \leq  \exp\left(\frac{P(\mathcal{A})\left(e^{H\lambda} - 1\right)}{H}\right),$$ where $P(\mathcal{A})$ is the expected total payoff of the algorithm.

We use $(w,k)$-dispersion to lower bound $W_{\numfunctions+1}/W_1$ in terms
  of the optimal parameter's total payout.  The key insight is that not only  $\vec{\rho}^*$  the optimal parameter gets a good payoff in hindsight, but all the parameters in the ball of radius $w$ around $\vec{\rho}^*$  have a good total payoff.
Let $\vec{\rho}^*$ be the optimal parameter and let $\OPT = U_{T+1}(\vec{\rho}^*)$. Also, let $\mathcal{B}^*$ be the ball of radius $w$ around $\vec{\rho}^*$. From $(w,k)$-dispersion, we know that for all $\vec{\rho} \in \mathcal{B}^*$, $$U_{T+1}(\vec{\rho}) \geq \OPT - Hk - LTw.$$ Therefore, \begin{align*}
W_{T+1} &= \int_\configs \exp(\lambda U_{T+1}(\vec{\rho})) \, d\vec{\rho} \geq \int_{\mathcal{B}^*} \exp(\lambda U_{T+1}(\vec{\rho})) \, d\vec{\rho}\\
&\geq \vol(B(\vec{\rho}^*, w))\exp(\lambda (\OPT - Hk - LTw)).
\end{align*}
Moreover, $W_1 = \int_{\configs} \exp(\lambda U_1(\vec{\rho})) \, d\vec{\rho} \leq \vol(B(\vec{0}, R))$. Therefore, \[\frac{W_{T+1}}{W_1} \geq \frac{\vol(B(\vec{\rho}^*, w))}{\vol(B(\vec{0}, R))} \exp(\lambda (\OPT - Hk - LTw)).\]

Combining the upper and lower bounds on $\frac{W_{T+1}}{W_1}$ gives the result.
\end{proof}

Whether the continuous weighted
 majority algorithm can be implemented in polynomial time depends on
 the setting. Assume that for all rounds $t \in \{1, \ldots, \numfunctions\}$,  $\sum_{s=1}^t u_s$ is
  piecewise Lipschitz over at most $N$ pieces. It is not hard to prove that when $d=1$ and $\configs=\reals$ and $\exp(\sum_{s=1}^t
  u_s)$ can be integrated in constant time on each of its pieces, the running
  time is $O(\numfunctions N)$ per round.
    When $d > 1$ and $\sum_{s=1}^t u_s$ is piecewise
  concave over convex pieces, ~\cite{BDV18} provide an efficient approximate
  implementation by using tools from high-dimensional  geometry.

We now show that under natural smoothness conditions about the input instances, dispersion is satisfied for the knapsack and clustering problems discussed in Section~\ref{se:dislearn}. The proof structure in both cases is to use the functional form of the discontinuities of the corresponding utility functions
  to reason about the distribution of discontinuity locations that arise as transformations of
random problem parameters in algorithm configuration instances. Using this idea one can upper bound the expected number of functions with discontinuities in any fixed interval, and then obtain the final desired result by using a uniform convergence result summarized in the following lemma.

%
\begin{lemma} \label{thm:dispersionTool}
  Let $u_{x_1}, u_{x_2}, \dots, u_{x_T} : \reals \to \reals$ be piecewise
  $L$-Lipschitz functions, each having at most $N$ discontinuities, with independent randomness in their discontinuities.\footnote{The independence is between functions. Within a function, the discontinuities may be correlated.}
  Let $\mathcal{F} = \{f_I : \Pbi \to \{0,1\} \mid I \subset
  \reals \text{ is an interval}\}$, where $f_I : \Pbi \to \{0,1\}$ maps an instance $x \in \Pbi$ to $1$ if the interval
  $I$ contains a discontinuity for the utility function $u_x$, and $0$
  otherwise.
 With probability $1-\delta$
  over randomness in the selection of utility functions $u_{x_1}, u_{x_2}, \dots, u_{x_T}$  we have:
  $$
    \sup_{f_I \in \mathcal{F}} \left|
      \sum_{t=1}^T f_I(x_t)
      - \expect\left[
        \sum_{t=1}^T f_I(x_t)
      \right]
    \right| \leq O(\sqrt{T \log(N/\delta)}).
  $$

\end{lemma}

Intuitively, Lemma \ref{thm:dispersionTool} states the following.  Suppose that instead of a worst-case sequence of utility functions, there is some randomness in the locations of their discontinuities, where the randomness is independent between utility functions.  Then, with high probability, for every interval $I$, the actual number of discontinuities in $I$ will be close to its expectation, and in particular within an additive gap of at most $O(\sqrt{T \log(N/\delta)})$.  To prove this, the key step is to apply uniform convergence to the class of functions $ \mathcal{F} $ defined in the lemma, and to prove that its VC-dimension is  $O(\log N)$.  This lemma is from \cite{BDP20} and it improves over the earlier result in~\cite{BDV18}.

For the   family of greedy algorithms $\algfamily^{knapsack}$ discussed in Section \ref{sec:clusteringdistributional}, in the worst case,  the associated utility functions might not be dispersed.  However, we can perform a smoothed analysis
and show that if there is some randomness in the item values, then we have dispersion with high probability.  Formally, we assume item values are $b$-smooth: they are random, independent, and each has a density function upper bounded by $b$.  (For example, a canonical $b$-smooth distribution is a uniform distribution on an interval of width $1/b$.)  The dispersion guarantee is given in Theorem \ref{thm:knapsackDispersion} below.

  \begin{theorem} \label{thm:knapsackDispersion}
  Let $x_1, \dots, x_T$ be any sequence of knapsack instances with $n$ items and
  capacity $C$ where instance $i$ has sizes $s^{(i)}_1, \dots, s^{(i)}_n \in
  [1,C]$ and values $v^{(i)}_1, \dots, v^{(i)}_n \in (0,1]$. Assume that the
  item values are $b$-smooth. Then for any $\delta > 0$, with
  probability at least $1-\delta$, for any $w > 0$, the utility functions
  $u_{x_1}, \dots, u_{x_T}$ are $(w,k)$-dispersed for
$  k = O\left(w T n^2 b^2 \log(C) + \sqrt{T \log(n/\delta)}\right).$
\end{theorem}
\begin{proof} [Proof Sketch]
  Recall from Lemma~\ref{knapsack} that for a
  knapsack instance $x$ with item values $v_1, \dots, v_n$ and sizes $s_1, \dots, s_n$,
  the discontinuities of the utility $u_x$ only occur at parameter values where
  the relative ordering of two items swaps under the score $\sigma_\rho$. For
  items $i$ and $j$, let $c_{ij} = \log(v_i / v_j) / \log(s_i / s_j)$ be the
  critical parameter value where their relative scores swap. When the item
  values are independent and have $b$-bounded distributions, we are
  guaranteed that their joint density is also $b^2$-bounded. Using this~\cite{BDV18} prove that each  discontinuity is random and has a density function that is upper bounded by  $b^2 \log(C) / 2$.

%
 Next, fix
  any ball $I$ of radius $w$ (i.e., an interval of width $2w$). For any function
  $u_{x_i}$ the probability that any one of its discontinuities belongs to the
  interval $I$ is at most $w b^2 \log(C)$. Summing over both knapsack
  instances $x_1, \dots, x_T$ and the $O(n^2)$ discontinuities for each, it
  follows that the expected total number of discontinuities in interval $I$ is
  at most $w T n^2 b^2 \log(C)$. This is also a bound on the expected
  number of functions among $u_{x_1}, \dots, u_{x_T}$ that are discontinuous on
  the ball $I$.
Finally,  Lemma~\ref{thm:dispersionTool}
can be used to show
  that with probability $\geq 1-\delta$ any interval of radius $w$ has 
  $O(wTn^2 b^2 \log(C) + \sqrt{T \log(n/\delta)})$ discontinuous functions.
\end{proof}

Combining the dispersion analysis for the knapsack problem with the regret
guarantees for the continuous weighted majority algorithm, we can obtain an upper bound on the algorithm's expected regret.
In particular, applying Theorems~\ref{thm:knapsackDispersion} and~\ref{thm:1_d_online} with $\delta = 1/\sqrt{T}$ and $w = 1/(\sqrt{T} n^2 b^2 \log(C))$, and using that
 utilities take values in $[0,C]$, we have the following corollary (from~\citet{BDP20} which improves over the earlier result~\citep{BDV18}):

\begin{corollary} \label{cor:knapsackRegret}
  Let $x_1, \dots, x_T$ be any sequence of knapsack instances satisfying the
  same conditions as in Theorem~\ref{thm:knapsackDispersion}. The continuous weighted majority algorithm employed to
  choose parameters $\rho_1, \dots, \rho_T
  \in [0,R]$ for the sequence $x_1, \dots, x_T$ with $\lambda = \sqrt{\log(R /
  (\sqrt{T} n^2 b^2 \log(C))} / C$ has expected regret bounded by
  \[
  \expect\left[
    \max_{\rho \in [0,R]} \sum_{t=1}^T u_{x_t}(\rho) - \sum_{t=1}^T u_{x_t}(\rho_t)
  \right]
  = O\left(C\sqrt{T \log(RTn b\log(C))}\right).
  \]
\end{corollary}

For the $\algfamily^{scl}$-linkage algorithm family analyzed in Section \ref{sec:clusteringdistributional},~\citet{BDP20}  show  the following  guarantees:

\begin{theorem}~\label{cor:linkageRegret}
  Let $x_1, \dots, x_T$  be a sequence of clustering instances over $n$ points and  let
  $D_1, \dots, D_T \in [0,M]^{n \times n}$ be their corresponding distance
  matrices. Assume that the
  pairwise distances for each instance are $b$-smooth: for all $t \in \{1, \ldots, T\}$
   the entries
  of $D_t$ are random, independent, and have
  density functions that are bounded by $b$.  Assume further that the utility functions
  are bounded in $[0,H]$. The continuous weighted majority algorithm employed to
  choose parameters $\rho_1, \dots, \rho_T \in [0,1]$ for the sequence $D_1,
  \dots, D_T$ with $\lambda = \sqrt{\log(1 / (\sqrt{T} n^8 b^2 M^2)} / H$
  has expected regret bounded by
  $$
  \expect\left[
  \max_{\rho \in [0,R]} \sum_{t=1}^T u_{x_t}(\rho)
  -
  \sum_{t=1}^T u_{x_t}(\rho_t)
  \right]
  = O\left(H\sqrt{T \log(T n b M)}\right). $$
\end{theorem}
\begin{proof} [(Proof Sketch)]
Using Lemma~\ref{lem:single} and properties of   $b$-bounded random variables we can show that  for any $\delta > 0$,
  with probability $1-\delta$, for any $w > 0$, the utility functions
  $u_{D_1}, \dots, u_{D_T}$ are $(w,k)$-dispersed for
  $ k = O\left(
  w T n^8 b^2 M^2 + \sqrt{T \log(n/\delta)}
  \right).
  $
  By choosing  $w =1/(\sqrt{T} n^8 b^2 M^2)$ and using Theorem~\ref{thm:1_d_online}, we obtain the  result.
\end{proof}

\cite{BDV18} also show good dispersion bounds for the family of algorithms for solving IQPs
discussed in  Section~\ref{se:dislearn} (SDP relaxations followed by parametrized rounding); interestingly, here  the dispersion condition holds due to internal randomization in the algorithms themselves (with no additional smoothness assumptions about the input instances).

\paragraph{Extensions} Extensions to the results presented here include:
\begin{itemize}

\item A regret bound of $\tilde
O(\numfunctions^{(d+1)/(d+2)}(H\sqrt{d(3R)^d} + L))$ for the bandit setting. While this is more realistic in terms of feedback per round,  the regret bound is significantly worse than that for the full information setting~\citep{BDV18}.

\item A better regret bound, similar to that of Theorem \ref{thm:1_d_online}, and a
computationally efficient implementation for semi-bandit online
optimization problems where evaluating the cost function of one
algorithm reveals the cost for a range of similar algorithms~\citep{BDP20}.

\item Results  for the shifting experts setting, showing that we can compete not only with the best fixed parameter setting in hindsight, but with a tougher benchmark, namely withthe  best  strategy  in  hindsight  that  can  shift  a  limited  number  of  times  between parameter settings, allowing our algorithm to adapt to changing environments~\cite{BDD20}.

\item An application of the dispersion condition to the offline
learning setting of Section~\ref{se:dislearn}, specifically the derivation of more refined
data-dependent uniform convergence guarantees using empirical Radamacher complexity~\citep{BDV18}.
\end{itemize}

\section{Summary and Discussion}
The results in Section~\ref{se:dislearn} showing a pseudo-dimension of $O(\log n)$ also lead to computationally efficient learning algorithms, because one can identify and try out a polynomial number of parameter choices. Those with larger pseudo-dimension generally need to use additional problem structure to achieve polynomial-time optimization and learning.

\paragraph{Other Directions}
Other recent theoretical works~\citep{WGS18,KLL17} consider data-driven algorithm design for runtime among a finite number of algorithms. Their goal is to select  an algorithm whose expected running time, after removing a $\delta$ probability mass of instances, is at most $(1 + \epsilon)OPT$, where $OPT$ is the expected running time of the best of the $n$ algorithms on instances from $D$. The key challenge they address is to minimize the total running time for learning in terms of $n$, $OPT$, $\epsilon$, and $\delta$, and without any dependence on the maximum runtime of an algorithm. Note that as opposed to most of the work presented in this chapter these papers do not assume any assume structural relations among the $n$ algorithms. It would be very interesting to combine these two lines of work.

A related line of work presents sample complexity bounds derived for data-driven mechanism  design for revenue maximization settings~\citep{TJ15,BSV16,BSV18}.
The general theorem of~\cite{BDDKSV19} (mentioned in Section~\ref{generic}) can be used to recover the  bounds in these papers.
Furthermore, the dispersion tools derived in  Section~\ref{se:online} have been used for providing estimators for the degree of approximate incentive compatibility of an auction, another important problem in modern auction design~\citep{BSV19}.

\paragraph{Open Directions}
Data driven algorithm design has the potential to fundamentally shift the way we analyze and design algorithms for combinatorial problems.  In addition to scaling up the techniques developed so far and also using them for new problems, it would be interesting to develop new analysis frameworks that lead to even better automated algorithm design techniques.  For example, it would be interesting to explore a reinforcement learning approach, where we would learn state-based decision policies that use properties of the current state of the algorithm (e.g., the search tree for MIPs) to determine how to proceed (e.g., which variable to branch on next at a given node).
It would also be interesting to develop tools for learning within a single problem instance (as opposed to learning across instances).

In addition to providing theoretically sound and practically useful data-driven algorithmic methods, in the long term, this  area has the potential to give rise to new algorithmic paradigms of the type humans were not able to design before.

\bibliographystyle{abbrvnat}
\bibliography{chap29}

\end{document}